\documentclass[11pt,reqno]{amsart}
\usepackage{txfonts}
\usepackage{mathrsfs}
\usepackage{amsfonts}
\allowdisplaybreaks[4]
\usepackage{graphicx} 
\usepackage{epsfig}
\usepackage{bbm}
\usepackage{amssymb}
\usepackage{amsmath,xypic}
\usepackage{cite,color}
\usepackage{extarrows}
\usepackage{tikz}
\usepackage{float}
\usepackage{booktabs}
\usepackage{tcolorbox}
\usepackage{colortbl}

\newtheorem{theorem}{Theorem}

\newtheorem{proposition}[theorem]{Proposition}

\newtheorem{remark}[theorem]{Remark}
\newtheorem{lemma}[theorem]{Lemma}
\newtheorem{example}{Example}[section]
\newcommand{\be}{\begin{equation}}
\newcommand{\ee}{\end{equation}}
\newcommand{\bea}{\begin{eqnarray}}
\newcommand{\eea}{\end{eqnarray}}
\newcommand{\ba}{\begin{array}}
	\newcommand{\ea}{\end{array}}
\newcommand{\bean}{\begin{eqnarray*}}
	\newcommand{\eean}{\end{eqnarray*}}

\newcommand{\La}{\Lambda}

\newcommand{\pa}{\partial}

\begin{document}

\title{Generalized Bigraded Toda Hierarchy}
\author{Yue Liu$^1$, Xingjie Yan$^{1,2}$, Jinbiao Wang$^1$, Jipeng Cheng$^{1,2*}$ }
\dedicatory { $^1$ School of Mathematics, China University of
Mining and Technology, Xuzhou, Jiangsu 221116, China\\
$^{2}$ Jiangsu Center for Applied Mathematics (CUMT), \ Xuzhou, Jiangsu 221116, China}
\thanks{*Corresponding author. Email: chengjp@cumt.edu.cn, chengjipeng1983@163.com.}
\begin{abstract}
Bigraded Toda hierarchy $L_1^M(n)=L_2^N(n)$ is generalized to $L_1^M(n)=L_2^{N}(n)+\sum_{j\in \mathbb Z}\sum_{i=1}^{m}q^{(i)}_n\Lambda^jr^{(i)}_{n+1}$, which is the
analogue of the famous constrained KP hierarchy $L^{k}=(L^{k})_{\geq0}+\sum_{i=1}^{m}q_{i}\partial^{-1}r_i$. It is known that different bosonizations of fermionic KP hierarchy will give rise to different kinds of integrable hierarchies. Starting from the fermionic form of constrained KP hierarchy, bilinear equation of this generalized bigraded Toda hierarchy (GBTH) are derived by using 2--component boson--fermion correspondence. Next based upon this, the Lax structure of GBTH is obtained. Conversely, we also derive bilinear equation of GBTH from the corresponding Lax structure.\\


\end{abstract}

\maketitle

\section{Introduction}
\subsection{Bosonizations of fermionic KP hierarchy}
Sato theory has gained great success in mathematical physics and integrable systems (see \cite{Willox2004,Date1983,Jimbo1983,Mulase1994,Ohta1988} and their references). Among various objects in Sato theory, KP hierarchy is most basic one. Here we would like to discuss the fermionic KP hierarchy \cite{Alexandrov2013,Kac2023,Kac2003,Date1983,Jimbo1983,Kac1998,Miwa2000}\  given by
\begin{align}\label{eq:fermionic KP hierarchy eq1}
\sum_{j\in \mathbb{Z}+1/2}\psi_j^{+}\tau\otimes \psi_{-j}^-\tau=0,\ \tau\in \mathcal{F},
\end{align}
where $\psi_j^\pm(j\in \mathbb{Z}+1/2)$\ are charged free fermions satisfying
\begin{align*}
\psi _j^\lambda\psi_l^\mu+\psi_l^\mu\psi _j^\lambda=\delta_{\lambda+\mu, 0}\delta_{j+l,0},\quad
\lambda,\mu =\pm\ \text{and}\
j,l\in \mathbb{Z}+1/2.
\end{align*}
And $\mathcal{F}$ is the fermionic Fock space spanned by
$
\psi^+_{i_1}\cdots\psi^+_{i_r}\psi^-_{j_1}\cdots\psi^-_{j_s}|0\rangle,\
$$
\ i_1<\cdots <i_r<0,\ j_1<\cdots<j_s<0,\
$
where vacuum vector\ $|0\rangle$\ is defined by
$\psi^\pm_s|0\rangle=0,\ s>0.$
If define charge of $\psi^\pm_{s}=\pm1$, then
$$
\mathcal{F}=\mathop{\bigoplus}\limits_{s\in\mathbb{Z}}\mathcal{F}_s,\quad
\mathcal{F}_s=\big\{a|0\rangle \in \mathcal{F}|\ \text{charge of}\ a=s\big\}.$$\
For $\tau\in \mathcal{F}$\ satisfying fermionic KP hierarchy \eqref{eq:fermionic KP hierarchy eq1}, it can be proved that there exists unique $\mathcal{F}_s$\ such that $\tau\in\mathcal{F}_s$\ (one can refer to Appendix \ref{appA} for more details). By different bosonizations of $\mathcal{F}$, fermionic KP hierarchy will become different kinds of integrable hierarchies \cite{Kac2023,Kac2003,Jimbo1983,Kac1998}, for example, usual KP hierarchy, Toda hierarchy and multi--component KP hierarchy. Next we will state two typical bosonizations \cite{Kac2023,Kac2003,Jimbo1983,Kac1998} of fermionic KP hierarchy.\\

\begin{tcolorbox}[title = {\textbf{Bosonization} \uppercase\expandafter{\romannumeral1}\ :}]
$
\sigma_1 :\mathcal{F}\cong\mathcal{B}_1
=\mathbb{C}[w,w^{-1},t=(t_1=x,t_2,t_3,\cdots )]
$
\end{tcolorbox}
\noindent Here $\sigma _1$\ is defined by
\ $\sigma_1(|0\rangle)=1$\ and
\begin{align*}
\sigma_1\psi^{\pm}(z)\sigma_1^{-1}
=w^{\pm1}z^{\pm w\frac{\partial}{\partial w}}
e^{\pm \xi(t,z)}e^{\mp \xi(\tilde{\pa}_t,z^{-1})},
\end{align*}
where\ $\xi(t,z)=\sum_{j=1}^{\infty}t_j z^j$,
\ $\tilde{\pa}_t=(\pa_{t_1},\pa_{t_2}/2,\pa_{t_3}/3,\cdots)$\ and
$$
\psi^{\pm}(z)=\sum_{k\in \mathbb{Z}+1/2}\psi^\pm_kz^{-k-1/2}.
$$
For $\tau$\ in fermionic KP hierarchy \eqref{eq:fermionic KP hierarchy eq1},\
if assume $\tau\in \mathcal{F}_s$ and set $\sigma_1(\tau)=w^s\tau(t),$\ we can get \cite{Kac2003,Jimbo1983}
 $${\rm Res}_z\tau(t-[z^{-1}])\tau(t'+[z^{-1}])e^{\xi(t-t',z)}=0,$$
where ${\rm Res}_z\sum_i a_iz^i=a_{-1}$ and
$[z^{-1}]=(z^{-1},z^{-2}/2,z^{-3}/3,\cdots)$.
This is the usual KP hierarchy having the following Lax equation \cite{Willox2004,Date1983,Ohta1988}
\begin{align*}
L_{t_p}=[(L^p)_{\geq0},L], \quad p=1,2,\cdots,
\end{align*}
with Lax operator $L=\pa +\sum_{i=1}^\infty u_i\pa^{-i}$ ($\pa=\pa_{t_1}$)
and $\big(\sum_{i}a_i\pa^i\big)_{\geq0}=\sum_{i\geq0}a_i\pa^i$,
where $u_i$ can be expressed by tau function $\tau(t)$\ according to
${\rm Res}_\pa L^p=\pa_x\pa_{t_p}{\log}\tau$
with
${\rm Res }_\pa\sum_{i}a_i\pa^i=a_{-1}$.

In order to express another bosonization, we can rearrange labels of charged free fermions $\psi^\pm_j$\ in the way below \cite{Kac2023,Kac2003,Kac1998,Jimbo1983},
\begin{align}
&\psi^{\pm(1)}_{-Mi-1/2-p}=\psi^\pm_{-(M+N)i-1/2-p},\quad
1\leq p\leq M,\label{eq:rearrange free fermions1}
\\
&\psi^{\pm(2)}_{-Ni-1/2-q}=\psi^\pm_{-(M+N)i-1/2-(M+q)},\quad
1\leq q\leq N,\label{eq:rearrange free fermions2}
\end{align}
where $M,N\geq1$\ and $i\in \mathbb{Z}$,\ then
\begin{align*}
\psi_s^{\lambda(a)}\psi_r^{\mu(b)}+\psi_s^{\mu(b)}\psi_r^{\lambda(a)}
=\delta_{\lambda+\mu,0}\delta_{r+s,0}\delta_{a,b},\quad
s,r\in\mathbb{Z}+1/2,\  \lambda,\mu=+,-\ \text{and}\ a,b=1,2.
\end{align*}
Similarly, we can define generating functions of $\psi_j^{\pm (a)}$\ as follows,
$$\psi^{\pm (a)} (z)=\sum_{j\in \mathbb{Z}+1/2}\psi_j^{\pm (a)} z^{-j-\frac{1}{2}},\quad a=1,2.$$

\begin{tcolorbox}[title = {\textbf{Bosonization}\ \uppercase\expandafter{\romannumeral2}\ :}]
$
\sigma_2 :\mathcal{F}\cong\mathcal{B}_2
=\bigoplus\limits_{l_1,l_2\in \mathbb{Z}}\mathbb{C}Q_1^{l_1}Q_2^{l_2}\otimes \mathbb{C}[\mathbf{t}=(t^{(1)},t^{(2)})]\
$
\end{tcolorbox}
\noindent Here
 $t^{(a)}=(t^{(a)}_1,t^{(a)}_2,\cdots)\ (a=1,2)$\ and
$Q_1Q_2=-Q_2Q_1.$\
And $\sigma _2$\ is defined by
\ $\sigma_2(|0\rangle)=1$,
\begin{align*}
\sigma _2\psi^{\pm(a)}(z)\sigma _2^{-1}
=Q_a^{\pm 1} z^{\pm Q_a\partial Q_a}{\rm exp}\big(\pm \xi(t^{(a)},z)\big){\rm exp}\big(\mp\xi(\tilde{\partial}_{t^{(a)}},z^{-1})\big).
\end{align*}
For $\tau\in \mathcal{F}_s$, if denote
$$\sigma_2(\tau)
=\sum_{n\in\mathbb{Z}}(-1)^{\frac{n(n-1)}{2}}Q_1^{n+s}Q_2^{-n}\tau_{n}(\mathbf{t})
,$$
we have
\begin{align}
&\oint_{C_R}\frac{dz}{2\pi {\bf i}}
\tau_n(\mathbf{t}-[z^{-1}]_1)\tau_{n'}(\mathbf{t}'+[z^{-1}]_1)
z^{n-n'}e^{\xi(t^{(1)}-t^{(1)'},z)}\nonumber\\
=&\oint_{C_r}\frac{dz}{2\pi {\bf i}}
\tau_{n+1}(\mathbf{t}-[z]_2)\tau_{n'-1}(\mathbf{t}'+[z]_2)
z^{n'-n}e^{\xi(t^{(2)}-t^{(2)'},z^{-1})},\label{eq:todabilinear}
\end{align}
which is just 2--Toda hierarchy \cite{Ueno1984,Alexandrov2013,Takasaki2018,Jimbo1983}. Here $\mathbf{t}\pm[z^{-1}]_1=(t^{(1)}\pm[z^{-1}],t^{(2)})$,  $\mathbf{t}\pm[z^{-1}]_2=(t^{(1)},t^{(2)}\pm[z^{-1}])$, and $C_R$ means the circle $|z|=R$ for sufficient large R, while $C_r$ is the circle $|z|=r$ with  sufficient small $r$. Both $C_R$ and $C_r$ are anticlockwise. Then the corresponding Lax equation \cite{Takasaki2018,Ueno1984} is
\begin{align}
\pa_{t^{(1)}_p}L_a(n)=[(L_1^p(n))_{\geq0},L_a(n)],\ \pa_{t^{(2)}_p}L_a(n)=[(L_2^p(n))_{<0},L_a(n)],\quad a=1,2.\label{eq:todalax}
\end{align}
Here Lax operators\ $L_1(n),L_2(n)$ are given by
\begin{align*}
L_1(n)&=\Lambda +a_0+a_1\Lambda^{-1}+\cdots, \\
L_2(n)&=b_{-1}\Lambda^{-1}+b_0+b_1\Lambda+\cdots,
\end{align*}
where $\Lambda$\ is shift operator acting on $f(n)$\ by
$\Lambda(f(n))=f(n+1)$,\
$\big(\sum_{i}a_i\La^i\big)_{P}=\sum_{P}a_i\La^i$ with $P\in\{\geq l,>l,\leq l,<l\}$
and $a_i=a_i(n,\mathbf{t})$, $b_i=b_i(n,\mathbf{t})$.

Now let us give a summary here.\ By Bosonizations \uppercase\expandafter{\romannumeral1}\ and \uppercase\expandafter{\romannumeral2},\  fermionic KP hierarchy
\eqref{eq:fermionic KP hierarchy eq1}\
is corresponding to KP and Toda hierarchies,\ respectively, which is shown in the following figure.

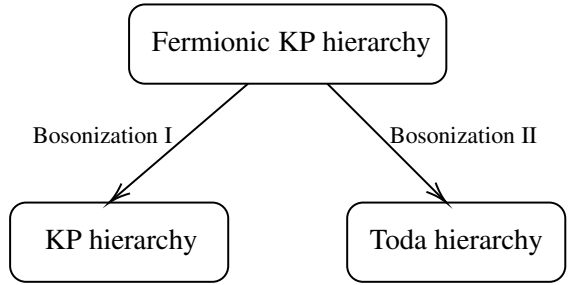
\begin{center}

\tikzset{every picture/.style={line width=0.75pt}}

\begin{tikzpicture}[x=0.75pt,y=0.75pt,yscale=-1,xscale=1]
\draw    (260,50) -- (318.59,108.59) ;
\draw [shift={(320,110)}, rotate = 225] [color={rgb, 255:red, 0; green, 0; blue, 0 }  ][line width=0.75]    (10.93,-3.29) .. controls (6.95,-1.4) and (3.31,-0.3) .. (0,0) .. controls (3.31,0.3) and (6.95,1.4) .. (10.93,3.29)   ;
\draw    (220,50) -- (151.52,108.7) ;
\draw [shift={(150,110)}, rotate = 319.4] [color={rgb, 255:red, 0; green, 0; blue, 0 }  ][line width=0.75]    (10.93,-3.29) .. controls (6.95,-1.4) and (3.31,-0.3) .. (0,0) .. controls (3.31,0.3) and (6.95,1.4) .. (10.93,3.29)   ;
\draw   (160,18) .. controls (160,13.58) and (163.58,10) .. (168,10) -- (312,10) .. controls (316.42,10) and (320,13.58) .. (320,18) -- (320,42) .. controls (320,46.42) and (316.42,50) .. (312,50) -- (168,50) .. controls (163.58,50) and (160,46.42) .. (160,42) -- cycle ;
\draw   (100,118) .. controls (100,113.58) and (103.58,110) .. (108,110) -- (202,110) .. controls (206.42,110) and (210,113.58) .. (210,118) -- (210,142) .. controls (210,146.42) and (206.42,150) .. (202,150) -- (108,150) .. controls (103.58,150) and (100,146.42) .. (100,142) -- cycle ;
\draw   (270,118) .. controls (270,113.58) and (273.58,110) .. (278,110) -- (372,110) .. controls (376.42,110) and (380,113.58) .. (380,118) -- (380,142) .. controls (380,146.42) and (376.42,150) .. (372,150) -- (278,150) .. controls (273.58,150) and (270,146.42) .. (270,142) -- cycle ;

\draw (116,122) node [anchor=north west][inner sep=0.75pt]   [align=left] {KP hierarchy};
\draw (280,122) node [anchor=north west][inner sep=0.75pt]   [align=left] {Toda hierarchy};
\draw (110,70) node [anchor=north west][inner sep=0.75pt]   [align=left] {\footnotesize Bosonization \uppercase\expandafter{\romannumeral1}};
\draw (290,70) node [anchor=north west][inner sep=0.75pt]   [align=left] {\footnotesize Bosonization \uppercase\expandafter{\romannumeral2}};
\draw (170,22) node [anchor=north west][inner sep=0.75pt]   [align=left] {Fermionic KP hierarchy};
\end{tikzpicture}
\end{center}

\subsection{Generalized bigraded Toda hierarchy}

Constrained KP (cKP) hierarchy \cite{Cheng1992,Cheng&Zhang1994} is one of the most important reductions of KP hierarchy,
which is defined by the Lax operator,
\begin{align}
&L^{k}=(L^{k})_{\geq0}+\sum_{i=1}^{m}q_{i}\partial^{-1}r_i,\label{eq:constrained KP Lax formula eq1}\\
&L_{t_p}=[(L^p)_{\geq0},L],\label{eq:constrained KP Lax formula eq2}\\
&q_{i,t_p}=(L^p)_{\geq0}(q_i),\label{eq:constrained KP Lax formula eq3}\\
&r_{i,t_p}=-(L^p)_{\geq0}^{*}(r_i).\label{eq:constrained KP Lax formula eq4}
\end{align}
To indicate the independence on $k$\ and $m$,
the system of
\eqref{eq:constrained KP Lax formula eq1}--\eqref{eq:constrained KP Lax formula eq4}\
is called\
$(k,m)$--cKP hierarchy. The cKP hierarchy \cite{Cheng1992,Cheng&Zhang1994} is related with  AKNS, Yajima--Oikawa and Melnikov hierarchies by choosing different $k$ and $m$.

In \cite{Cheng&Zhang1994}, $(k,m)$--cKP hierarchy
is equivalent to
\begin{align}
&\sum_{i=1}^{m}q_i(t)r_i(t')={\rm Res}_zz^k\psi^+(t,z)\psi^-(t',z),
\label{eq:ckp-equivalent-eq1}\\
&q_i(t)={\rm Res}_zz^{-1}\psi^+(t,z)\psi^-(t',z)q_i(t'+[z^{-1}]),
\label{eq:ckp-equivalent-eq2}\\
&r_i(t)={\rm Res}_zz^{-1}\psi^-(t,z)\psi^+(t',z)r_i(t'-[z^{-1}]),
\label{eq:ckp-equivalent-eq3}
\end{align}
where $\psi^{\pm}(t,z)$\ are KP wave functions related with KP tau function
$\tau_0(t)$\ by
\begin{align*}
\psi^+(t,z)=\frac{\tau_0(t-[z^{-1}])}{\tau_0(t)}e^{\xi(t,z)},
\quad \psi^-(t,z)=\frac{\tau_0(t+[z^{-1}])}{\tau_0(t)}e^{-\xi(t,z)}.
\end{align*}
Further if introduce
$$\tau_{1}^{(i)}(t)= q_i(t)\tau_0(t),\quad \tau_{-1}^{(i)}(t)=r_i(t)\tau_0(t),$$
then
\eqref{eq:ckp-equivalent-eq1}--\eqref{eq:ckp-equivalent-eq3}
will become
\begin{align}
&{\rm Res}_zz^k\tau_0(t-[z^{-1}])\tau_0(t{'}+[z^{-1}])e^{\xi(t-t{'},z)}
=\sum_{i=1}^{m}\tau_{1}^{(i)}(t)\tau_{-1}^{(i)}(t{'}),\label{eq:bilinear-eq1}\\
&{\rm Res}_zz^{-1}\tau_0(t-[z^{-1}])\tau_{1}^{(i)}(t{'}+[z^{-1}])e^{\xi(t-t{'},z)}
=\tau_{1}^{(i)}(t)\tau_0(t{'}),\label{eq:bilinear-eq2}\\
&{\rm Res}_zz^{-1}\tau_{-1}^{(i)}(t-[z^{-1}])\tau_0(t{'}+[z^{-1}])e^{\xi(t-t{'},z)}
=\tau_0(t)\tau_{-1}^{(i)}(t{'}).\label{eq:bilinear-eq3}
\end{align}
If we use the inverse map $\sigma_1^{-1}$ of Bosonization \uppercase\expandafter{\romannumeral1}\
and set
$$\tau_0=\sigma_1^{-1}\big(w^s\tau_0(t)\big)\in \mathcal{F}_s,\quad \tau_l^{(i)}=\sigma_1^{-1}\big(w^{s+l}\tau_l^{(i)}(t)\big)\in \mathcal{F}_{s+l},\ l=\pm 1,$$
we can get
\begin{align}
&\Omega_{(k)}(\tau_0\otimes\tau_0)
=\sum_{i=1}^{m}\tau_{1}^{(i)}\otimes\tau_{-1}^{(i)},\label{eq:fermionic bilinear equation1}\\
&\Omega_{(0)}(\tau_0\otimes\tau_{1}^{(i)})
=\tau_{1}^{(i)}\otimes\tau_0,\label{eq:fermionic bilinear equation2}\\			&\Omega_{(0)}(\tau_{-1}^{(i)}\otimes\tau_0)
=\tau_0\otimes\tau_{-1}^{(i)},\label{eq:fermionic bilinear equation3}
\end{align}
which is called the fermionic $(k,m)$--cKP hierarchy. Here
\begin{align}
 \Omega  _{(s)}=\sum_{i\in \mathbb{Z}+1/2}\psi_i^{+}\otimes \psi_{-i+s}^-={\rm Res}_zz^s\psi^+(z)\otimes \psi^-(z).\label{eq:casimir operator}
 \end{align}

In this paper we will apply
Bosonization \uppercase\expandafter{\romannumeral2}\ to equations
\eqref{eq:fermionic bilinear equation1}--\eqref{eq:fermionic bilinear equation3}\
and set\ $k=M+N$,\ then we will get bilinear equations of Toda hierarchy with some constraints
and further obtain the corresponding Lax structure as follows.
\begin{align}
&L_1^M(n)=L_2^{N}(n)+\sum_{j\in \mathbb Z}\sum_{i=1}^{m}q^{(i)}_n\Lambda^jr^{(i)}_{n+1},\quad M,\ N\geq1,
\label{eq:ckp-lax-operator1}
\\
&\pa_{t^{(1)}_p}L_a(n)=[(L_1^p(n))_{\geq0},L_a(n)],\quad\ \pa_{t^{(2)}_p}L_a(n)=[(L_2^p(n))_{<0},L_a(n)],\quad a=1,2,
\label{eq:ckp-lax-operator2}
\\
&\partial_{ t_{p}^{(1)}}q^{(i)}_n
=\Big(L_1^{p}(n)\Big)_{\geq0}\Big(q^{(i)}_n\Big),\quad\quad\quad\
\partial _{ t_{p}^{(2)}}q^{(i)}_n=\Big(L_2^{p}(n)\Big)_{<0}\Big(q^{(i)}_n\Big),\label{eq:ckp-lax-operator3}\\
&\partial_{ t_{p}^{(1)}}r^{(i)}_n
=-\Big(\big(L_1^{p}(n-1)\big)_{\geq0}\Big)^*\left(r^{(i)}_n\right),\quad
\partial _{ t_{p}^{(2)}}r^{(i)}_n=-\Big(\big(L_2^{p}(n-1)\big)_{<0}\Big)^*\left(r^{(i)}_n\right),\label{eq:ckp-lax-operator4}
\end{align}
where $\big(\sum_{i}a_i\La^i\big)^*=\sum_{i}\La^{-i}a_i$.

Note that \eqref{eq:ckp-lax-operator1} is one reduction of Toda hierarchy, which is a kind of generalization for bigraded Toda hierarchy \cite{Ueno1984,Takasaki2018,Carlet2006}. And this reduction of Toda hierarchy in \eqref{eq:ckp-lax-operator1} should be related with the work in \cite{Konopelchenko1992}, where the Toda reduction is constructed by symmetry constraints. The suitable name for this reduction should be constrained Toda hierarchy, since it comes from the fermionic constrained KP hierarchy. But in \cite{Krichever2022}, constrained Toda hierarchy means another different reduction of Toda hierarchy, which is quite similar to CKP hierarchy\cite{Zabrodin2021,Krichever2021}. So it will be better to call the system \eqref{eq:ckp-lax-operator1}--\eqref{eq:ckp-lax-operator4} to be $(M,N,m)$--generalized bigraded Toda hierarchy (GBTH). Besides above reductions of Toda hierarchy, others contain B--Toda \& C--Toda \cite{Ueno1984}, periodic Toda \cite{Takasaki2018,Ueno1984}, Ablowitz--Ladik hierarchy \cite{Takasaki2018}, Toda hierarchy with constraint of type B \cite{Krichever2023} (also called large BKP hierarchy \cite{Guan}) and so on.

Let us summarize above process to obtain GBTH in the figure below.
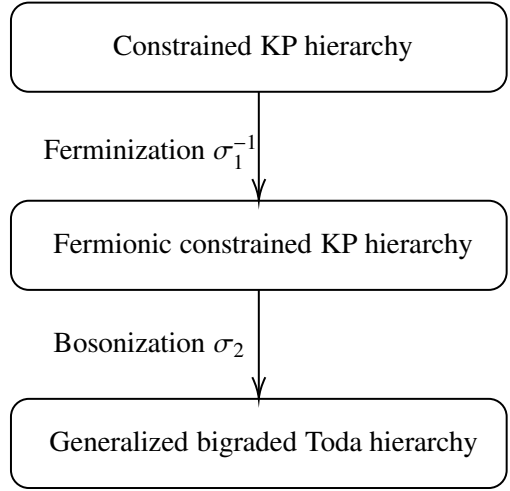
\begin{center}
\tikzset{every picture/.style={line width=0.75pt}}

\begin{tikzpicture}[x=0.75pt,y=0.75pt,yscale=-1,xscale=1]

\draw   (90,279) .. controls (90,274.03) and (94.03,270) .. (99,270) -- (331,270) .. controls (335.97,270) and (340,274.03) .. (340,279) -- (340,306) .. controls (340,310.97) and (335.97,315) .. (331,315) -- (99,315) .. controls (94.03,315) and (90,310.97) .. (90,306) -- cycle ;
\draw   (90,379) .. controls (90,374.03) and (94.03,370) .. (99,370) -- (331,370) .. controls (335.97,370) and (340,374.03) .. (340,379) -- (340,406) .. controls (340,410.97) and (335.97,415) .. (331,415) -- (99,415) .. controls (94.03,415) and (90,410.97) .. (90,406) -- cycle ;
\draw   (90,479) .. controls (90,474.03) and (94.03,470) .. (99,470) -- (331,470) .. controls (335.97,470) and (340,474.03) .. (340,479) -- (340,506) .. controls (340,510.97) and (335.97,515) .. (331,515) -- (99,515) .. controls (94.03,515) and (90,510.97) .. (90,506) -- cycle ;
\draw    (215,315) -- (215,355.85) -- (215,368) ;
\draw [shift={(215,370)}, rotate = 270] [color={rgb, 255:red, 0; green, 0; blue, 0 }  ][line width=0.75]    (10.93,-3.29) .. controls (6.95,-1.4) and (3.31,-0.3) .. (0,0) .. controls (3.31,0.3) and (6.95,1.4) .. (10.93,3.29)   ;
\draw    (215,415) -- (215,468) ;
\draw [shift={(215,470)}, rotate = 270] [color={rgb, 255:red, 0; green, 0; blue, 0 }  ][line width=0.75]    (10.93,-3.29) .. controls (6.95,-1.4) and (3.31,-0.3) .. (0,0) .. controls (3.31,0.3) and (6.95,1.4) .. (10.93,3.29)   ;

\draw (140,285) node [anchor=north west][inner sep=0.75pt]   [align=left] {Constrained KP hierarchy};
\draw (110,385) node [anchor=north west][inner sep=0.75pt]   [align=left] {Fermionic constrained KP hierarchy};
\draw (108,485) node [anchor=north west][inner sep=0.75pt]   [align=left] {Generalized bigraded Toda hierarchy};
\draw (110,435) node [anchor=north west][inner sep=0.75pt]   [align=left] {Bosonization\ $\sigma _2$};
\draw (105,335) node [anchor=north west][inner sep=0.75pt]   [align=left] {Ferminization $\sigma_1^{-1}$};

\end{tikzpicture}

\end{center}

This paper is organized in the way below. In Section 2, we firstly apply Bosonization II to fermionic constrained KP hierarchy and obtain the corresponding bilinear equations. Then the Lax structure is derived from the bilinear equations. In Section 3, we show the equivalence of bilinear equations and Lax structure for GBTH by deriving bilinear equations from Lax structure. Finally, some conclusions and discussions are given in Section 4.

\section{From bilinear equations to Lax structure}
In this section, we will start from the fermionic $(k,m)$--cKP hierarchy
\eqref{eq:fermionic bilinear equation1}--\eqref{eq:fermionic bilinear equation3},
and use Bosonization \uppercase\expandafter{\romannumeral2}\ to obtain the bilinear equations for
$(M,N,m)$--GBTH. Then based upon this,  the corresponding Lax equations are obtained.

First of all, we can find from \eqref{eq:rearrange free fermions1}\eqref{eq:rearrange free fermions2}\ and \eqref{eq:casimir operator} that
\begin{align}
\Omega_{(M+N)}=\Omega^{(1)}_{(M)}+\Omega^{(2)}_{(N)},\quad M, N\geq 1,\label{lemma:Omegasum}
\end{align}
where$\ \Omega^{(a)} _{(s)}={\rm Res}_zz^s\psi^{+(a)}(z)\otimes \psi^{-(a)}(z),\ a=1,2$.
Therefore, \eqref{eq:fermionic bilinear equation1}--\eqref{eq:fermionic bilinear equation3}\
will become
\begin{align}
&{\rm Res}_zz^M\psi^{+(1)}(z)\tau_0\otimes \psi^{-(1)}(z)\tau_0
+{\rm Res}_zz^N\psi^{+(2)}(z)\tau_0\otimes \psi^{-(2)}(z)\tau_0
=\sum_{i=1}^{m}\tau_{1}^{(i)}\otimes \tau_{-1}^{(i)},
\label{eq:resfermionic bilinear equation1}\\
&{\rm Res}_z\psi^{+(1)}(z)\tau_0\otimes \psi^{-(1)}(z)\tau_{1}^{(i)}
+{\rm Res}_z\psi^{+(2)}(z)\tau_0\otimes \psi^{-(2)}(z)\tau_{1}^{(i)}
=\tau_{1}^{(i)}\otimes \tau_0,
\label{eq:resfermionic bilinear equation2}\\
&{\rm Res}_z\psi^{+(1)}(z)\tau_{-1}^{(i)}\otimes \psi^{-(1)}(z)\tau_0
+{\rm Res}_z\psi^{+(2)}(z)\tau_{-1}^{(i)}\otimes \psi^{-(2)}(z)\tau_0
=\tau_0\otimes \tau_{-1}^{(i)}.\label{eq:resfermionic bilinear equation3}
\end{align}

If apply\ $\sigma _2\otimes \sigma_2$\ to \eqref{eq:resfermionic bilinear equation1}, and denote
\begin{align*}
&\sigma_2(\tau_0)
=\sum_{n\in\mathbb{Z}}(-1)^{\frac{n(n-1)}{2}}Q_1^{n}Q_2^{-n}\tau_{0,n}(\mathbf{t}),\quad
\sigma_2(\tau_{\pm 1}^{(i)})
=\sum_{n\in\mathbb{Z}}(-1)^{\frac{n(n-1)}{2}}Q_1^{n\pm 1}Q_2^{-n}\tau^{(i)}_{\pm 1,n}(\mathbf{t}),\\
&Q_a\otimes 1=Q'_a,\quad 1\otimes Q_a=Q''_a,\quad
t^{(a)}_i\otimes 1=t^{(a)'}_i,
\quad 1\otimes t^{(a)}_i=t^{(a)''}_i,\quad a=1,2,
\end{align*}
then equation \eqref{eq:resfermionic bilinear equation1}\ can be written into
\begin{align*}
&\sum_{n',n''\in \mathbb{Z}}{\rm Res}_zz^{M+n'-n''}e^{\xi(t^{(1)'}-t^{(1)''},z)}
\tau_{0,n'}(\mathbf{t}'-[z^{-1}]_1)\tau_{0,n''}(\mathbf{t}''+[z^{-1}]_1)
Q_1'^{n'+1}Q_2'^{-n'}Q_1''^{n''-1}Q_2''^{-n''}\\
&-\sum_{n',n''\in \mathbb{Z}}{\rm Res}_zz^{N-n'+n''-2}e^{\xi(t^{(2)'}-t^{(2)''},z)}
\tau_{0,n'+1}(\mathbf{t}'-[z^{-1}]_2)\tau_{0,n''-1}(\mathbf{t}''+[z^{-1}]_2)
Q_1'^{n'+1}Q_2'^{-n'}Q_1''^{n''-1}Q_2''^{-n''}\\
=&\sum_{n',n''\in \mathbb{Z}}\sum_{i=1}^{m}
\tau_{1,n'}^{(i)}(\mathbf{t}')\tau_{-1,n''}^{(i)}(\mathbf{t}'')
Q_1'^{n'+1}Q_2'^{-n'}Q_1''^{n''-1}Q_2''^{-n''}.
\end{align*}
Comparing coefficients of\ $Q'^{n'+1}_1Q'^{-n'}_2Q''^{n''-1}_1Q''^{-n''}_2$,\ one can get
\begin{align}
&\oint_{C_R}\frac{dz}{2\pi {\bf i}}z^{M+n'-n''}
\tau_{0,n'}(\mathbf{t}'-[z^{-1}]_1)
\tau_{0,n''}(\mathbf{t}''+[z^{-1}]_1)
e^{\xi(t^{(1)'}-t^{(1)''},z)}\nonumber\\
=&\oint_{C_r}\frac{dz}{2\pi {\bf i}}z^{-N+n'-n''}
\tau_{0,n'+1}(\mathbf{t}'-[z]_2)
\tau_{0,n''-1}(\mathbf{t}''+[z]_2)
e^{\xi(t^{(2)'}-t^{(2)''},z^{-1})}
+\sum_{i=1}^{m}\tau_{1,n'}^{(i)}(\mathbf{t}')
\tau_{-1,n''}^{(i)}(\mathbf{t}'').\label{eq:propbosonicbilineareq1}
\end{align}
Similarly, we can get from \eqref{eq:resfermionic bilinear equation2}\ and
\eqref{eq:resfermionic bilinear equation3} that
\begin{align}
\bullet\quad&\oint_{C_R}\frac{dz}{2\pi {\bf i}}z^{n'-n''-1}
\tau_{0,n'}(\mathbf{t}'-[z^{-1}]_1)
\tau_{1,n''}^{(i)}(\mathbf{t}''+[z^{-1}]_1)
e^{\xi(t^{(1)'}-t^{(1)''},z)}\nonumber\\
=&-\oint_{C_r}\frac{dz}{2\pi {\bf i}}z^{n'-n''}
\tau_{0,n'+1}(\mathbf{t}'-[z]_2)
\tau_{1,n''-1}^{(i)}(\mathbf{t}''+[z]_2)
e^{\xi(t^{(2)'}-t^{(2)''},z^{-1})}+\tau_{1,n'}^{(i)}(\mathbf{t}')
\tau_{0,n''}(\mathbf{t}''),\label{eq:propbosonicbilineareq2}\\
\bullet\quad&\oint_{C_R}\frac{dz}{2\pi {\bf i}}z^{n'-n''-1}
\tau_{-1,n'}^{(i)}(\mathbf{t}'-[z^{-1}]_1)
\tau_{0,n''}(\mathbf{t}''+[z^{-1}]_1)
e^{\xi(t^{(1)'}-t^{(1)''},z)}\nonumber\\
=&-\oint_{C_r}\frac{dz}{2\pi {\bf i}}z^{n'-n''}
\tau_{-1,n'+1}^{(i)}(\mathbf{t}'-[z]_2)
\tau_{0,n''-1}(\mathbf{t}''+[z]_2)
e^{\xi(t^{(2)'}-t^{(2)''},z^{-1})}
+\tau_{0,n'}(\mathbf{t}')
\tau_{-1,n''}^{(i)}(\mathbf{t}'').\label{eq:propbosonicbilineareq3}
\end{align}

In terms of Hirota bilinear operator \cite{Hirota2004}
\begin{align*}
P(D)f(\mathbf{t})\cdot g(\mathbf{t})
=P(\pa_y)\Big(f(\mathbf{t}+\mathbf{y})g(\mathbf{t}-\mathbf{y})\Big)|_{\mathbf{y}=0},
\end{align*}
where $P(D)=P(D^{(1)},D^{(2)})$,\ $\pa_y=(\pa_{y^{(1)}},\pa_{y^{(2)}})$,\ $D^{(a)}=(D^{(a)}_1,D^{(a)}_2,\cdots)$,\ $y^{(a)}=(y^{(a)}_1,y^{(a)}_2,\cdots)$, we give some examples for \eqref{eq:propbosonicbilineareq1}--\eqref{eq:propbosonicbilineareq3} ($M=N=1$) as follows,
\begin{align*}
&\Big(D_1^{(1)}-D_1^{(2)}\Big)\tau_{0,n+1}\cdot \tau_{0,n}=\sum_{i=1}^m \tau_{1,n}^{(i)}\cdot\tau_{-1,n+1}^{(i)},\\ &\Big(D_2^{(1)}+D_1^{(1)}D_1^{(2)}\Big)\tau_{0,n+1}\cdot \tau_{0,n}=\sum_{i=1}^m D_1^{(1)}\tau_{1,n}^{(i)}\cdot\tau_{-1,n+1}^{(i)},\\
&D^{(1)}_{1}D^{(2)}_{1}\tau_{0,n+1}\cdot\tau_{1,n}^{(i)}
+D^{(2)}_{1}\tau^{(i)}_{1,n+1}\cdot\tau_{0,n}=0,\\ &D^{(1)}_{1}D^{(2)}_{1}\tau_{0,n}\cdot\tau_{1,n+1}^{(i)}
+D^{(2)}_{1}D^{(1)}_{2}\tau^{(i)}_{1,n}\cdot\tau_{0,n+1}=0,\\
&D^{(1)}_{1}D^{(2)}_{1}\tau_{-1,n+1}^{(i)}\cdot\tau_{0,n}
+D^{(2)}_{1}\tau_{0,n+1}\cdot\tau_{-1,n}^{(i)}=0,\\ &D^{(1)}_{1}D^{(2)}_{1}\tau_{-1,n}^{(i)}\cdot\tau_{0,n+1}
+D^{(2)}_{1}D^{(1)}_{2}\tau_{0,n}\cdot\tau_{-1,n+1}^{(i)}=0.
\end{align*}
\begin{remark}
In what follows, we will find that \eqref{eq:propbosonicbilineareq1} is corresponding to the constraint of Toda hierarchy \eqref{eq:ckp-lax-operator1}. While \eqref{eq:propbosonicbilineareq2} and \eqref{eq:propbosonicbilineareq3} are modified Toda hierarchy\cite{Guan}, also called 2--component 1st modified KP hierarchy \cite{Vandeleur2015,Jimbo1983}, that is,
\begin{align}
&\oint_{C_R}\frac{dz}{2\pi {\bf i}}z^{n-n'-1}
\tau_{0,n}(\mathbf{t}-[z^{-1}]_1)
\tau_{1,n'}(\mathbf{t}'+[z^{-1}]_1)
e^{\xi(t^{(1)}-t^{(1)'},z)}\nonumber\\
=&-\oint_{C_r}\frac{dz}{2\pi {\bf i}}z^{n-n'}
\tau_{0,n+1}(\mathbf{t}-[z]_2)
\tau_{1,n'-1}(\mathbf{t}'+[z]_2)
e^{\xi(t^{(2)}-t^{(2)'},z^{-1})}+\tau_{1,n}(\mathbf{t})
\tau_{0,n'}(\mathbf{t}').\label{mTodabilinear}
\end{align}
\end{remark}
\begin{lemma}\label{lemma:mToda}
Given $(\tau_{0,n},\tau_{1,n})$ satisfying modified Toda hierarchy \eqref{mTodabilinear},
\begin{align*}
&\tau_{1,n+1}({\mathbf t}+[z^{-1}]_1)\tau_{0,n}({\mathbf t})
-z\cdot\tau_{1,n}({\mathbf t}+[z^{-1}]_1)\tau_{0,n+1}({\mathbf t})
=-z\cdot\tau_{0,n+1}({\mathbf t}+[z^{-1}]_1)\tau_{1,n}({\mathbf t}),\\
&\tau_{1,n}({\mathbf t}+[z]_2)\tau_{0,n}({\mathbf t})
-z\cdot\tau_{1,n-1}({\mathbf t}+[z]_2)\tau_{0,n+1}({\mathbf t})
=\tau_{0,n}({\mathbf t}+[z]_2)\tau_{1,n}({\mathbf t}).
\end{align*}
\end{lemma}
\begin{proof}
The first relation can be obtained by setting $\mathbf{t}-\mathbf{t}'=[\lambda^{-1}]_1$ and $n'=n+1$ in \eqref{mTodabilinear}, while the second one is derived by letting $\mathbf{t}-\mathbf{t}'=[\lambda]_2$ and $n'=n$.
\end{proof}

Next if introduce wave function $\psi_i(n,\mathbf{t},z)$, adjoint wave function $\psi^*_i(n,\mathbf{t},z)$, eigenfunction $q^{(i)}_n(\mathbf{t})$ and adjoint eigenfunction $r^{(i)}_n(\mathbf{t})$ as follows,
\begin{align}
&\psi_1(n,\mathbf{t},z)
\triangleq\frac{\tau_{0,n}(\mathbf{t}-[z^{-1}]_1)}
{\tau_{0,n}(\mathbf{t})}z^ne^{\xi(t^{(1)},z)},\label{eq:todawavefunction1}\\
&\psi_2(n,\mathbf{t},z)
\triangleq\frac{\tau_{0,n+1}(\mathbf{t}-[z]_2)}
{\tau_{0,n}(\mathbf{t})}z^ne^{\xi(t^{(2)},z^{-1})},\label{eq:todawavefunction2}\\
&\psi_1^*(n,\mathbf{t},z)
\triangleq\frac{\tau_{0,n}(\mathbf{t}+[z^{-1}]_1)}
{\tau_{0,n}(\mathbf{t})}z^{-n+1}e^{-\xi(t^{(1)},z)},\label{eq:todawavefunction3}\\
&\psi_2^*(n,\mathbf{t},z)
\triangleq\frac{\tau_{0,n-1}(\mathbf{t}+[z]_2)}
{\tau_{0,n}(\mathbf{t})}z^{-n+1}e^{-\xi(t^{(2)},z^{-1})},\label{eq:todawavefunction4}\\
&q^{(i)}_n(\mathbf{t})=\frac{\tau_{1,n}^{(i)}(\mathbf{t})}{\tau_{0,n}(\mathbf{t})},\quad
r^{(i)}_n(\mathbf{t})=\frac{\tau_{-1,n}^{(i)}(\mathbf{t})}{\tau_{0,n}(\mathbf{t})},\label{eq:q,r-tau}
\end{align}
then \eqref{eq:propbosonicbilineareq1}--\eqref{eq:propbosonicbilineareq3} can be written into
\begin{align}
\bullet\quad&\oint_{C_R}\frac{dz}{2\pi {\bf i}}
z^{M-1}\psi_1(n',\mathbf{t}',z)\psi^*_1(n'',\mathbf{t}'',z)\nonumber\\
=&\oint _{C_r}\frac{dz}{2\pi {\bf i}}z^{-N-1}\psi_2(n',\mathbf{t}',z)\psi^*_2(n'',\mathbf{t}'',z)
+\sum_{i=1}^{m}q^{(i)}_{n'}(\mathbf{t}')r^{(i)}_{n''}(\mathbf{t}''),
\label{eq:bosoniceq1'}
\\
\bullet\quad&\oint_{C_R}\frac{dz}{2\pi {\bf i}}
z^{-2}\psi_1(n',\mathbf{t}',z)
q^{(i)}_{n''}(\mathbf{t}''+[z^{-1}]_1)
\psi_1^*(n'',\mathbf{t}'',z)\nonumber\\
=&-\oint_{C_r}\frac{dz}{2\pi {\bf i}}
z^{-1}\psi_2(n',\mathbf{t}',z)
q^{(i)}_{n''-1}(\mathbf{t}''+[z]_2)
\psi_2^*(n'',\mathbf{t}'',z)
+q^{(i)}_{n'}(\mathbf{t}'),\label{eq:bosoniceq2'}
\\
\bullet\quad&\oint _{C_R}\frac{dz}{2\pi {\bf i}}
z^{-2}r^{(i)}_{n'}(\mathbf{t}'-[z^{-1}]_1)\psi_1(n',\mathbf{t}',z)
\psi_1^*(n'',\mathbf{t}'',z)\nonumber\\
=&-\oint_{C_r}\frac{dz}{2\pi {\bf i}}
z^{-1}r^{(i)}_{n'+1}(\mathbf{t}'-[z]_2)\psi_2(n',\mathbf{t}',z)
\psi_2^*(n'',\mathbf{t}'',z)
+r^{(i)}_{n''}(\mathbf{t}'').\label{eq:bosoniceq3'}
\end{align}
Before further discussion, we need the lemma below by Lemma \ref{lemma:mToda}.
\begin{lemma}\label{lemma:lambda-qr}
Given $q_n^{(i)}$, $r_n^{(i)}$, $\psi_j$ and $\psi^*_j$ defined in
\eqref{eq:todawavefunction1}--\eqref{eq:q,r-tau},
\begin{align}
&q^{(i)}_{n}(\mathbf{t}+[z^{-1}]_1)
\psi_1^*(n,\mathbf{t},z)
=-z\cdot\iota_{\Lambda}(\Lambda -1)^{-1}
\Big(
\psi^*_1(n+1,\mathbf{t},z)\cdot q_{n}^{(i)}(\mathbf{t})
\Big),\label{eq:fay-qr1}
\\
&q^{(i)}_{n-1}(\mathbf{t}+[z]_2)
\psi_2^*(n,\mathbf{t},z)=
\iota_{\Lambda^{-1}}(\Lambda -1)^{-1}
\Big(
\psi_2^*(n+1,\mathbf{t},z)\cdot
q_{n}^{(i)}(\mathbf{t})
\Big),\label{eq:fay-qr2}
\\
&r^{(i)}_{n}(\mathbf{t}-[z^{-1}]_1)\psi_1(n,\mathbf{t},z)
=z\cdot\iota_{\Lambda^{-1}}(\Lambda -1)^{-1}
\Big(
\psi_1(n,\mathbf{t},z)\cdot r_{n+1}^{(i)}(\mathbf{t})
\Big),\label{eq:fay-qr3}
\\
&r^{(i)}_{n+1}(\mathbf{t}-[z]_2)\psi_2(n,\mathbf{t},z)
=-\iota_{\Lambda}(\Lambda-1)^{-1}
\Big(
\psi_2(n,\mathbf{t},z)\cdot r_{n+1}^{(i)}(\mathbf{t})
\Big),\label{eq:fay-qr4}
\end{align}
where right hand sides of above relations are computed by $\Big(\sum_i a_i\Lambda^i\Big)(z^{\pm n})=\sum_i a_iz^{\pm(i+n)}$ and $\iota_{\Lambda}(\Lambda-1)^{-1}=-\sum_{i\geq 0}\Lambda^i$ and $\iota_{\Lambda^{-1}}(\Lambda -1)^{-1}=\sum_{i\geq 1}\Lambda^{-i}$.
\end{lemma}
\begin{remark}
Recall that $\psi_j$ and $\psi_j^*$ have the form $A(z^{\pm n})$ for some operator $A=\sum_i a_i\Lambda^i$, respectively.
Then the choice of $\iota_{\Lambda}(\Lambda-1)^{-1}$ or $\iota_{\Lambda^{-1}}(\Lambda -1)^{-1}$ is determined by expansions of $z$. For instance in \eqref{eq:fay-qr1}, left hand side has the form of\ $z^{-n+1}\Big(q^{(i)}_{n}(\mathbf{t})+\mathcal{O}(z^{-1})\Big)e^{-\xi(t^{(1)},z)}
$, so we should take $\iota_{\Lambda}(\Lambda-1)^{-1}$ on the right hand side of \eqref{eq:fay-qr1} to keep the corresponding form.
\end{remark}
\begin{lemma}\label{lemma:UV}\cite{Adler1999}
Let $A(n,\Lambda)=\sum_{j}a_j(n)\Lambda^j$,
$B(n,\Lambda)=\sum_{j}b_j(n)\Lambda^j$ be two operators, then
\begin{align*}
A(n,\Lambda)B^*(n,\Lambda)
=\sum_{l\in \mathbb{Z}}{\rm Res}_zz^{-1}A(n,\Lambda)(z^{n})B(n+l,\Lambda)(z^{-n-l})\Lambda^l,
\end{align*}
by $(\sum_ia_i\La^i)(z^{-n})=\sum_ia_iz^{-(n+i)}$.
\end{lemma}
\begin{proposition}
Given  $q_n^{(i)}$, $r_n^{(i)}$, $\psi_j$ and $\psi_j^*$ defined in equations \eqref{eq:todawavefunction1}--\eqref{eq:q,r-tau}, we can find \eqref{eq:bosoniceq1'}--\eqref{eq:bosoniceq3'} become
\begin{align}
\bullet\quad\sum_{i=1}^mq_{n'}^{(i)}(\mathbf{t}&')
r_{n''}^{(i)}(\mathbf{t}'')
=\oint_{C_R}\frac{dz}{2\pi {\bf i}}z^{M-1}\psi_1(n',\mathbf{\mathbf{t}}',z)
\psi_1^*(n'',\mathbf{\mathbf{t}}'',z)\notag\\
&\quad \quad\quad\quad-\oint_{C_r}\frac{dz}{2\pi {\bf i}}z^{-N-1}\psi_2(n',\mathbf{t}',z)
\psi_2^*(n'',\mathbf{t}'',z),\label{eq:todaresbilinear1}\\
\bullet\quad q^{(i)}_{n'}(\mathbf{t}')
=&-\oint_{C_R}\frac{dz}{2\pi {\bf i}}z^{-1}\psi_1(n',\mathbf{t}',z)\iota_{\Lambda}(\Lambda-1)^{-1}
\Big(
\psi_1^*(n''+1,\mathbf{t}'',z)q^{(i)}_{n''}(\mathbf{t}'')
\Big)\notag\\
&+\oint_{C_r}\frac{dz}{2\pi {\bf i}}z^{-1}\psi_2(n',\mathbf{t}',z)\iota_{\Lambda^{-1}}(\Lambda-1)^{-1}
\Big(
\psi_2^*(n''+1,\mathbf{t}'',z)q^{(i)}_{n''}(\mathbf{t}'')
\Big),\label{eq:todaresbilinear2}\\
\bullet\quad r^{(i)}_{n'}(\mathbf{t}')
=&\oint_{C_R}\frac{dz}{2\pi {\bf i}}z^{-1}\psi_1^*(n',\mathbf{t}',z)\iota_{\Lambda^{-1}}(\Lambda-1)^{-1}
\Big(
\psi_1(n'',\mathbf{t}'',z)r^{(i)}_{n''+1}(\mathbf{t}'')
\Big)\notag\\
&-\oint_{C_r}\frac{dz}{2\pi {\bf i}}z^{-1}\psi_2^*(n',\mathbf{t}',z)\iota_{\Lambda}(\Lambda-1)^{-1}
\Big(
\psi_2(n'',\mathbf{t}'',z)r^{(i)}_{n''+1}(\mathbf{t}'')
\Big).\label{eq:todaresbilinear3}
\end{align}
\end{proposition}
After the preparation above, let us introduce operators $W_a,\ S_a\ (a=1,2)$ as follows,
\begin{align*}
&W_1(n,\mathbf{t},\Lambda)
=S_1(n,\mathbf{t},\Lambda)e^{\xi(t^{(1)},\Lambda)},\quad\ \ \
\widetilde{W}_1(n,\mathbf{t},\Lambda)
=\widetilde{S}_1(n,\mathbf{t},\Lambda)e^{-\xi(t^{(1)},\Lambda^{-1})},\\
&W_2(n,\mathbf{t},\Lambda)
=S_2(n,\mathbf{t},\Lambda)e^{\xi(t^{(2)},\Lambda^{-1})},\quad
\widetilde{W}_2(n,\mathbf{t},\Lambda)
=\widetilde{S}_2(n,\mathbf{t},\Lambda)e^{-\xi(t^{(2)},\Lambda)},
\end{align*}
with
\begin{align*}
&S_1(n,\mathbf{t},\Lambda)
=1+\sum_{j=1}^{+\infty}\omega_j\Lambda^{-j},\quad
\widetilde{S}_1(n,\mathbf{t},\Lambda)
=\Lambda^{-1}+\sum_{j=1}^{+\infty}\widetilde{\omega}_j\Lambda^{j-1},\\
&S_2(n,\mathbf{t},\Lambda)
=\nu_0+\sum_{j=1}^{+\infty}\nu_j\Lambda^{j},\quad\
\widetilde{S}_2(n,\mathbf{t},\Lambda)
=\sum_{j=1}^{+\infty}\widetilde{\nu}_j\Lambda^{-j-1},
\end{align*}
satisfying
$
\psi_a(n,\mathbf{t},z)=W_a(n,\mathbf{t},\Lambda)(z^{n}),\
\psi_a^*(n,\mathbf{t},z)=\widetilde{W}_a(n,\mathbf{t},\Lambda)(z^{-n}).
$
Next if define
\begin{align*}
L_1(n)=S_1\Lambda S_1^{-1}, \quad L_2(n)=S_2\Lambda^{-1}S_2^{-1},
\end{align*}
we have the following theorem.

\begin{theorem}
Given bilinear equations  \eqref{eq:todaresbilinear1}--\eqref{eq:todaresbilinear3}, we have
\begin{align}
&L_1^M(n)=L_2^{N}(n)+\sum_{j\in \mathbb Z}\sum_{i=1}^{m}q^{(i)}_n\Lambda^jr^{(i)}_{n+1},
\label{eq:ckp-lax-operator1'}\quad M,\ N\geq1,
\\
&\pa_{t^{(1)}_p}L_a(n)=[(L_1^p(n))_{\geq0},L_a(n)],\quad\ \pa_{t^{(2)}_p}L_a(n)=[(L_2^p(n))_{<0},L_a(n)],\quad a=1,2,
\label{eq:ckp-lax-operator2'}
\\
&\partial_{ t_{p}^{(1)}}q^{(i)}_n
=\Big(L_1^{p}(n)\Big)_{\geq0}\Big(q^{(i)}_n\Big),\quad\ \quad\quad
\partial _{ t_{p}^{(2)}}q^{(i)}_n=\Big(L_2^{p}(n)\Big)_{<0}\Big(q^{(i)}_n\Big),\label{eq:ckp-lax-operator3'}\\
&\partial_{ t_{p}^{(1)}}r^{(i)}_n
=-\Big(\big(L_1^{p}(n-1)\big)_{\geq0}\Big)^*\left(r^{(i)}_n\right),\quad
\partial _{ t_{p}^{(2)}}r^{(i)}_n=-\Big(\big(L_2^{p}(n-1)\big)_{<0}\Big)^*\left(r^{(i)}_n\right).\label{eq:ckp-lax-operator4'}
\end{align}

\begin{proof}
Firstly apply $\Lambda-1$\ on both sides of equation \eqref{eq:todaresbilinear2}\ with respect to $n''$,
\begin{align*}
\oint_{C_R}\frac{dz}{2\pi {\bf i}}z^{-1}\psi_1(n',\mathbf{t}',z)
\psi^*_1(n''+1,\mathbf{t}'',z)
=\oint_{C_r}\frac{dz}{2\pi {\bf i}}z^{-1}\psi_2(n',\mathbf{t}',z)
\psi^*_2(n''+1,\mathbf{t}'',z),
\end{align*}
which is bilinear equation of Toda hierarchy\cite{Ueno1984,Takasaki2018}. Therefore by Lemma \ref{lemma:UV}, we can get
$\widetilde{S}_a^*=S_a^{-1}\cdot \Lambda\ (a=1,2)$\ and
\begin{align}
&\pa_{t^{(1)}_p}S_1=-(S_1\Lambda^p S_1^{-1})_{<0}S_1,\ \
\pa_{t^{(2)}_p}S_1=(S_2\Lambda^{-p} S_2^{-1})_{<0}S_1,\label{eq:partialS1}\\
&\pa_{t^{(1)}_p}S_2=(S_1\Lambda^p S_1^{-1})_{\geq0}S_2, \ \ \ \
\pa_{t^{(2)}_p}S_2=-(S_2\Lambda^{-p} S_2^{-1})_{\geq0}S_2.\label{eq:partialS2}
\end{align}
Based upon these relations, one can easily obtain \eqref{eq:ckp-lax-operator2'}. Notice that by \eqref{eq:partialS1} and \eqref{eq:partialS2}, we have
\begin{align*}
&\partial _{ t_{p}^{(1)}}\psi_a(n,{\bf t},z)=\Big(L_1^p(n)\Big)_{\geq0}\Big(\psi_a(n,{\bf t},z)\Big),\quad
\partial _{ t_{p}^{(1)}}\psi_a^*(n,{\bf t},z)=-\Big(\big(L_1^p(n-1)\big)_{\geq0}\Big)^*\Big(\psi_a^*(n,{\bf t},z)\Big),\\
&\partial _{ t_{p}^{(2)}}\psi_a(n,{\bf t},z)=\Big(L_2^p(n)\Big)_{\geq0}\Big(\psi_a(n,{\bf t},z)\Big),\quad
\partial _{ t_{p}^{(2)}}\psi_a^*(n,{\bf t},z)=-\Big(\big(L_2^p(n-1)\big)_{<0}\Big)^*\Big(\psi_a^*(n,{\bf t},z)\Big).
\end{align*}
Based upon above, we can know that \eqref{eq:ckp-lax-operator3'} and \eqref{eq:ckp-lax-operator4'} hold by \eqref{eq:todaresbilinear2},
\eqref{eq:todaresbilinear3}.

Finally by Lemma \ref{lemma:UV}, \eqref{eq:todaresbilinear1}\ can be written into
\begin{align*}
&S_1(n,\mathbf{t}',\Lambda)\Lambda^Me^{\xi_1(t'-t'',\Lambda)}
\widetilde{S}_1^*(n,\mathbf{t}'',\Lambda)\\
=&S_2(n,\mathbf{t}',\Lambda)\Lambda^{-N}e^{\xi_2(t'-t'',\Lambda^{-1})}
\widetilde{S}_2^*(n,\mathbf{t}'',\Lambda)
+\sum_{j\in \mathbb{Z}}\sum_{i=1}^{m}q^{(i)}_{n}(\mathbf{t}')r^{(i)}_{n+j}(\mathbf{t}'')\Lambda^{j}.
\end{align*}
After substituting $\widetilde{S}_a^*=S_a^{-1}\cdot \Lambda\ (a=1,2)$\ into the above equation and letting ${\mathbf t}'={\mathbf t}''={\mathbf t}$,
\begin{align*}
S_1(n,\mathbf{t},\Lambda)\Lambda^MS_1(n,\mathbf{t},\Lambda)^{-1}=S_2(n,\mathbf{t},\Lambda)\Lambda^{-N}S_2(n,\mathbf{t},\Lambda)^{-1}
+\sum_{j\in \mathbb{Z}}\sum_{i=1}^{m}q^{(i)}_{n}(\mathbf{t})r^{(i)}_{n+j}(\mathbf{t})\Lambda^{j-1},
\end{align*}
which implies
\eqref{eq:ckp-lax-operator1'}\ holds.
\end{proof}
\end{theorem}
\begin{remark}
Notice that \eqref{eq:ckp-lax-operator1'} is equivalent to the following three relations
\begin{align}
&\Big(L_1^M(n)\Big)_{\leq -N-1}=\sum_{i=1}^{m}q^{(i)}_n\cdot \Lambda^{-N}\Delta^{-1} \cdot r^{(i)}_{n+1},\label{L1mdkp}\\
&\Big(L_2^{N}(n)\Big)_{\geq M+1}=-\sum_{i=1}^{m}q^{(i)}_n\cdot \Lambda^M\Delta^{*-1}\cdot r^{(i)}_{n+1},\label{L2mdkp}\\
&\Big(L_1^M(n)\Big)_{\geq -N}=\Big(L_2^{N}(n)\Big)_{\leq M}+\sum_{j=-N}^M\sum_{i=1}^{m}q^{(i)}_n\Lambda^jr^{(i)}_{n+1},\label{L12mdkp}
\end{align}
where $\Delta=\Lambda-1$, $\Delta^{-1}=\sum_{j\geq 1}\Lambda^{-j}$,  $\Delta^{*-1}=\sum_{j\geq 1}\Lambda^{j}$. Here \eqref{L1mdkp} can be viewed as the constraint of discrete KP hierarchy in \cite{Oevel1996}. If assume $$
L_1^M(n)=\Lambda^M +\sum_{j=-\infty}^{M-1}\widetilde{a}_j\Lambda^{j},
\quad
L_2^N(n)=\sum_{j=-\infty}^{N}\widetilde{b}_j\Lambda^{-j},$$
then we can find that only $M+N+2$ functions:
$$\widetilde{a}_i(n) (0\leq i\leq M-1), \quad \widetilde{b}_i(n) (1\leq i\leq N), \quad q_n,\quad r_n$$
are independent in the whole system of $(M,N,m)$--GBTH \eqref{eq:ckp-lax-operator1'}--\eqref{eq:ckp-lax-operator4'}.
\end{remark}
\begin{example}
When $M=N=m=1$,\ \eqref{eq:ckp-lax-operator1'}\ becomes
\begin{align*}
L_1(n)=L_2(n)+\sum_{j\in \mathbb{Z}}q_n\Lambda^{j}r_{n+1}.
\end{align*}
Recall that $L_1(n)=\Lambda+\sum_{i=1}^\infty a_i(n)\Lambda^{-i}$ and\ $L_2(n)=\sum_{i=-1}^\infty b_i(n)\Lambda^i$. Then we can find
\begin{align*}
&b_1(n)=1-q_nr_{n+2},\quad
b_0(n)=a_0(n)-q_nr_{n+1},\quad
b_{-1}(n)=a_1(n)-q_nr_n,\\
& a_{j+1}(n)=q_nr_{n-j},\quad
b_{j+1}(n)=-q_nr_{n+j+2},\quad j\geq1.
\end{align*}
So in this case only
$q_n,\ r_n,\ a_0(n),\ b_{-1}(n)$ are independent in $(1,1,1)$--GBTH.
Further by \eqref{eq:ckp-lax-operator2'},\  \eqref{eq:ckp-lax-operator3'} and \eqref{eq:ckp-lax-operator4'}, we can get
\begin{align*}
&\partial_{ t_{1}^{(1)}}q_n
=q_{n+1}+a_0(n)q_n,\quad
\partial_{ t_{1}^{(1)}}a_0
=b_{-1}(n+1)-b_{-1}(n)+q_{n+1}r_{n+1}-q_{n}r_{n},\\
&\partial _{t_{1}^{(1)}}r_n
=-r_{n-1}-a_0(n-1)r_n,
\quad
\partial_{ t_{1}^{(1)}}b_{-1}(n)
=b_{-1}(n)\Big(a_0(n)-a_0(n-1)\Big),\\
&\partial_{ t_{1}^{(2)}}q_n
=b_{-1}(n)q_{n-1},
\quad
\partial_{ t_{1}^{(2)}}a_0
=b_{-1}(n)-b_{-1}(n+1),\\
&\partial_{ t_{1}^{(2)}}r_n
=-b_{-1}(n)r_{n+1},
\quad
\partial_{ t_{1}^{(2)}}b_{-1}(n)
=b_{-1}(n)\Big(a_0(n-1)-a_0(n)+q_nr_{n+1}-q_{n-1}r_{n}\Big).
\end{align*}
Notice that when $b_{-1}(n)=-q_n r_n$, we can find above relations are consistent with (7.8)--(7.12) in \cite{Konopelchenko1992}.
\end{example}

\section{From Lax structure to bilinear equations}
In this section, we will start from Lax structure of $(M,N,m)$--GBTH below and derive the corresponding bilinear equations \eqref{eq:todaresbilinear1}--\eqref{eq:todaresbilinear3}.
\begin{align}
&L_1^M(n)=L_2^{N}(n)+\sum_{j\in \mathbb Z}\sum_{i=1}^{m}q^{(i)}_n\Lambda^jr^{(i)}_{n+1},\quad M,\  N\geq 1,
\label{eq:ckp-lax-operator1''}
\\
&\pa_{t^{(1)}_p}L_a(n)=[(L_1^p(n))_{\geq0},L_a(n)],\quad\ \pa_{t^{(2)}_p}L_a(n)=[(L_2^p(n))_{<0},L_a(n)],\quad a=1,2,
\label{eq:ckp-lax-operator2''}
\\
&\partial_{ t_{p}^{(1)}}q^{(i)}_n
=\Big(L_1^{p}(n)\Big)_{\geq0}\Big(q^{(i)}_n\Big),\quad\quad\ \quad
\partial _{ t_{p}^{(2)}}q^{(i)}_n=\Big(L_2^{p}(n)\Big)_{<0}\Big(q^{(i)}_n\Big),\label{eq:ckp-lax-operator3''}\\
&\partial_{ t_{p}^{(1)}}r^{(i)}_n
=-\Big(\big(L_1^{p}(n-1)\big)_{\geq0}\Big)^*\left(r^{(i)}_n\right),\quad
\partial _{ t_{p}^{(2)}}r^{(i)}_n=-\Big(\big(L_2^{p}(n-1)\big)_{<0}\Big)^*\left(r^{(i)}_n\right).\label{eq:ckp-lax-operator4''}
\end{align}
Notice that this system is well defined, since\
$
\partial_{t_p}\Big(L_1^M(n)-L_2^N(n)\Big)
=\partial_{t_p}\Big(\sum_{j\in \mathbb Z}\sum_{i=1}^{m}q^{(i)}_n\Lambda^jr^{(i)}_{n+1}\Big),
$\
which can be derived by following formulas for $A=\sum_{j}a_j(n)\Lambda^j$,
\begin{align*}
A\cdot \sum_{l\in \mathbb{Z}}f\Lambda ^lg
=\sum_{l\in \mathbb{Z}}A(f)\cdot \Lambda ^l\cdot g,\quad
\sum_{l\in \mathbb{Z}}f\Lambda ^lg\cdot A
=\sum_{l\in \mathbb{Z}}f\cdot\Lambda ^l\cdot A^*(g).
\end{align*}

Firstly let us prove  \eqref{eq:todaresbilinear2}. Notice that
$
f(n+j,{\mathbf t})=\Lambda^j\Big(f(n,{\mathbf t})\Big)
$\
and
\begin{align*}
&\partial_{ t_{p}^{(1)}}q^{(i)}_n
=\Big(L_1^{p}(n)\Big)_{\geq0}\Big(q^{(i)}_n\Big),\quad
\partial _{ t_{p}^{(1)}}\psi_a(n,{\bf t},z)=\Big(L_1^p(n)\Big)_{\geq0}\Big(\psi_a(n,{\bf t},z)\Big),
\\
&\partial _{ t_{p}^{(2)}}q^{(i)}_n
=\Big(L_2^{p}(n)\Big)_{<0}\Big(q^{(i)}_n\Big),
\quad
\partial _{ t_{p}^{(2)}}\psi_a(n,{\bf t},z)=\Big(L_2^p(n)\Big)_{\geq0}\Big(\psi_a(n,{\bf t},z)\Big),\quad a=1,2.
\end{align*}
Therefore there exists $A_{\alpha}(n,\mathbf{t}, \Lambda)=\sum_{p=-s_\alpha}^{l_\alpha}a_{\alpha,p}(n,\mathbf{t})\Lambda^p$ such that
\begin{align*}
&q^{(i)}_{n+j}({\mathbf t}')
=\sum_{\alpha\geq0}\frac{({\mathbf t}'-{\mathbf t})^\alpha}{\alpha!}A_\alpha (n,\mathbf{t}, \Lambda)\Big(q^{(i)}_{n}({\mathbf t})\Big),\\
&\psi_a(n+j,\mathbf{t}',z)
=\sum_{\alpha_{\geq 0}}\frac{({\mathbf t}'-{\mathbf t})^\alpha}{\alpha!}A_\alpha (n,\mathbf{t}, \Lambda)\Big(\psi_a(n,\mathbf{t},z)\Big),
\quad a=1,2,
\end{align*}
where $\alpha=(\alpha^{(1)},\alpha^{(2)})$,\  $\mathbf{t}^\alpha=(t^{(1)})^{\alpha^{(1)}}(t^{(2)})^{\alpha^{(2)}}$ and $\alpha^{(b)}=(\alpha^{(b)}_1,\alpha^{(b)}_2,\cdots)$,\
$(t^{(b)})^{\alpha^{(b)}}=\Pi_{l=1}^{+\infty}(t^{(b)})^{\alpha_l^{(b)}},\ b=1,2$.
Hence if \eqref{eq:todaresbilinear2} holds when ${\mathbf t}'={\mathbf t}$,  then \eqref{eq:todaresbilinear2}\ can also hold for general ${\mathbf t}$ and ${\mathbf t}'$. Now let us concentrate on the proof of case ${\mathbf t}'={\mathbf t}$\ in \eqref{eq:todaresbilinear2}, that is,
\begin{align}\label{eq:todaresbilinear2equivalent2}
q^{(i)}_{n}({\mathbf t})
=\oint_{C_R}\frac{dz}{2\pi {\bf i}}
z^{-1}\psi_1(n,{\mathbf t},z)\rho _1(n+j,{\mathbf t},z)
+\oint_{C_r}\frac{dz}{2\pi {\bf i}}
z^{-1}\psi_2(n,{\mathbf t},z)\rho _2(n+j,{\mathbf t},z),\quad j\in\mathbb{Z},
\end{align}
where
$$
\rho_1(n,{\mathbf t},z)
=-\iota_{\Lambda^{-1}}(\Lambda-1)^{-1}\Big(\psi_1^*(n+1,\mathbf{t},z)q^{(i)}_{n}(\mathbf{t})\Big),\ \rho_2(n,{\mathbf t},z)
=\iota_{\Lambda}(\Lambda-1)^{-1}
\Big(\psi_2^*(n+1,\mathbf{t},z)q^{(i)}_{n}(\mathbf{t})\Big).\
$$
Further notice that \eqref{eq:todaresbilinear2equivalent2} is equivalent to
\begin{align}\label{eq:todaresbilinear2equivalent3}
q^{(i)}_{n}({\mathbf t})\sum_{j\in \mathbb{Z}}\Lambda^j
=\sum_{j\in \mathbb{Z}}\oint_{C_R}\frac{dz}{2\pi {\bf i}}
z^{-1}\psi_1(n,{\mathbf t},z)\rho _1(n+j,{\mathbf t},z)\Lambda^j
+\sum_{j\in \mathbb{Z}}\oint_{C_r}\frac{dz}{2\pi {\bf i}}
z^{-1}\psi_2(n,{\mathbf t},z)\rho _2(n+j,{\mathbf t},z)\Lambda^j.
\end{align}
So by Lemma \ref{lemma:UV} and
\begin{align*}
&\psi_1(n,{\mathbf t},z)
=\Big(S_1(n,{\mathbf t},\Lambda)e^{\xi(t^{(1)},\Lambda)}\Big)(z^n),\quad \psi_1^*(n+1,{\mathbf t},z)
=\Big(S_1^{-1}(n,{\mathbf t},\Lambda)^*e^{-\xi(t^{(1)},\Lambda^{-1})}\Big)(z^{-n}),\\
&\psi_2(n,{\mathbf t},z)
=\Big(S_2(n,{\mathbf t},\Lambda)e^{\xi(t^{(2)},\Lambda^{-1})}\Big)(z^{n}),\quad \psi_2^*(n+1,{\mathbf t},z)
=\Big(S_2^{-1}(n,{\mathbf t},\Lambda)^*e^{-\xi(t^{(2)},\Lambda)}\Big)(z^{-n}),
\end{align*}
we have
\begin{align*}
\text {RHS of \eqref{eq:todaresbilinear2equivalent3}}=q^{(i)}_{n}({\mathbf t})\Big(-\iota_{\Lambda^{-1}}(\Lambda^{-1}-1)^{-1}
+\iota_{\Lambda}(\Lambda^{-1}-1)^{-1}\Big)=\text {LHS of \eqref{eq:todaresbilinear2equivalent3}}.
\end{align*}
Similarly, we can prove \eqref{eq:todaresbilinear3}.

Now let us prove \eqref{eq:todaresbilinear1}.  Notice that \eqref{eq:ckp-lax-operator2''}\ is Lax equation of Toda hierarchy, so the wave function $\psi_a$ and adjoint wave function $\psi^*_a$ satisfy the following bilinear equation,
\begin{align}
\oint_{C_R}\frac{dz}{2\pi {\bf i}}z^{-1}\psi_1(n,\mathbf{t},z)
\psi^*_1(n',\mathbf{t}',z)
=\oint_{C_r}\frac{dz}{2\pi {\bf i}}z^{-1}\psi_2(n,\mathbf{t},z)
\psi^*_2(n',\mathbf{t}',z).\label{eq:Todabilinearwave}
\end{align}
Recall that
$$
L_1(n,\mathbf{t},\Lambda)\Big(\psi_1(n,\mathbf{t},z)\Big)=z\cdot\psi_1(n,\mathbf{t},z),\  L_2(n,\mathbf{t},\Lambda)\Big(\psi_2(n,\mathbf{t},z)\Big)=z^{-1}\cdot\psi_2(n,\mathbf{t},z),
$$
so if applying
$$L_1^M(n,\mathbf{t},\Lambda)=L_2^{N}(n,\mathbf{t},\Lambda)+\sum_{j\in \mathbb Z}\sum_{i=1}^{m}q^{(i)}_n(\mathbf{t})\Lambda^jr^{(i)}_{n+1}(\mathbf{t})$$ on  \eqref{eq:Todabilinearwave}, we can get
\begin{align}
&\oint_{C_R}\frac{dz}{2\pi {\bf i}}z^{M-1}\psi_1(n,\mathbf{t},z)\psi^*_1(n',\mathbf{t}',z)
-\oint_{C_r}\frac{dz}{2\pi {\bf i}}z^{-N-1}\psi_2(n,\mathbf{t},z)
\psi^*_2(n',\mathbf{t}',z)\nonumber\\
=&\sum_{j\in \mathbb Z}\sum_{i=1}^{m}
q^{(i)}_n(\mathbf{t})r^{(i)}_{n+j+1}(\mathbf{t})
\oint_{C_r}\frac{dz}{2\pi {\bf i}}z^{-1}
\psi_2(n+j,\mathbf{t},z)
\psi^*_2(n',\mathbf{t}',z).\label{eq:Todabilinearwave'}
\end{align}
If denote right hand side of \eqref{eq:Todabilinearwave'}\ as $A(n,n',{\mathbf t},{\mathbf t}')$, then by \eqref{eq:Todabilinearwave}
\begin{align*}
A(n,n',{\mathbf t},{\mathbf t}')
=&\sum_{j= -\infty}^{-1}\sum_{i=1}^{m}
q^{(i)}_n(\mathbf{t})r^{(i)}_{n+j+1}(\mathbf{t})
\oint_{C_R}\frac{dz}{2\pi {\bf i}}z^{-1}
\psi_1(n+j,\mathbf{t},z)
\psi^*_1(n',\mathbf{t}',z)\\
&+\sum_{j= 0}^{+\infty}\sum_{i=1}^{m}
q^{(i)}_n(\mathbf{t})r^{(i)}_{n+j+1}(\mathbf{t})
\oint_{C_r}\frac{dz}{2\pi {\bf i}}z^{-1}
\psi_2(n+j,\mathbf{t},z)
\psi^*_2(n',\mathbf{t}',z).
\end{align*}
Next according to
$\iota_{\Lambda}(\Lambda-1)^{-1}=-\sum_{j=0}^{+\infty}\Lambda^j,\
\iota_{\Lambda^{-1}}(\Lambda-1)^{-1}=\sum_{j=1}^{+\infty}\Lambda^{-j},
$ we can find by  \eqref{eq:todaresbilinear3}
\begin{align*}
A(n,n',{\mathbf t},{\mathbf t}')
=&\sum_{i=1}^{m}q^{(i)}_n(\mathbf{t})
\oint_{C_R}\frac{dz}{2\pi {\bf i}}z^{-1}\psi_1^*(n',\mathbf{t}',z)\iota_{\Lambda^{-1}}(\Lambda-1)^{-1}
\Big(
\psi_1(n,\mathbf{t},z)r^{(i)}_{n+1}(\mathbf{t})
\Big)\nonumber\\
&-\sum_{i=1}^{m}q^{(i)}_n(\mathbf{t})\oint_{C_r}\frac{dz}{2\pi {\bf i}}z^{-1}\psi_2^*(n',\mathbf{t}',z)\iota_{\Lambda}(\Lambda-1)^{-1}
\Big(
\psi_2(n,\mathbf{t},z^{-1})r^{(i)}_{n+1}(\mathbf{t})
\Big)\\
=&\sum_{i=1}^{m}q^{(i)}_n(\mathbf{t})r^{(i)}_{n'}(\mathbf{t}'),
\end{align*}
which means\ \eqref{eq:todaresbilinear1}\ holds.

Finally, let us summarize above results into the following theorem.
\begin{theorem}
Given Lax structure of GBTH \eqref{eq:ckp-lax-operator1''}--\eqref{eq:ckp-lax-operator4''},
\begin{align*}
\bullet\quad\sum_{i=1}^mq_{n'}^{(i)}(\mathbf{t}&')
r_{n''}^{(i)}(\mathbf{t}'')
=\oint_{C_R}\frac{dz}{2\pi {\bf i}}z^{M-1}\psi_1(n',\mathbf{\mathbf{t}}',z)
\psi_1^*(n'',\mathbf{\mathbf{t}}'',z)\\
&\quad \quad\quad\quad-\oint_{C_r}\frac{dz}{2\pi {\bf i}}z^{-N-1}\psi_2(n',\mathbf{t}',z)
\psi_2^*(n'',\mathbf{t}'',z),\\
\bullet\quad q^{(i)}_{n'}(\mathbf{t}')
=&-\oint_{C_R}\frac{dz}{2\pi {\bf i}}z^{-1}\psi_1(n',\mathbf{t}',z)\iota_{\Lambda}(\Lambda-1)^{-1}
\Big(
\psi_1^*(n''+1,\mathbf{t}'',z)q^{(i)}_{n''}(\mathbf{t}'')
\Big)\\
&+\oint_{C_r}\frac{dz}{2\pi {\bf i}}z^{-1}\psi_2(n',\mathbf{t}',z)\iota_{\Lambda^{-1}}(\Lambda-1)^{-1}
\Big(
\psi_2^*(n''+1,\mathbf{t}'',z)q^{(i)}_{n''}(\mathbf{t}'')
\Big),\\
\bullet\quad r^{(i)}_{n'}(\mathbf{t}')
=&\oint_{C_R}\frac{dz}{2\pi {\bf i}}z^{-1}\psi_1^*(n',\mathbf{t}',z)\iota_{\Lambda^{-1}}(\Lambda-1)^{-1}
\Big(
\psi_1(n'',\mathbf{t}'',z)r^{(i)}_{n''+1}(\mathbf{t}'')
\Big)\\
&-\oint_{C_r}\frac{dz}{2\pi {\bf i}}z^{-1}\psi_2^*(n',\mathbf{t}',z)\iota_{\Lambda}(\Lambda-1)^{-1}
\Big(
\psi_2(n'',\mathbf{t}'',z)r^{(i)}_{n''+1}(\mathbf{t}'')
\Big).
\end{align*}

\end{theorem}

\section{Conclusions and Discussions}
Starting from the constrained KP hierarchy $L^{k}=(L^{k})_{\geq0}+\sum_{i=1}^{m}q_{i}\partial^{-1}r_i$,
we firstly obtain its fermionic form
$\Omega_{(k)}(\tau_0\otimes\tau_0)
=\sum_{i=1}^{m}\tau_{1}^{(i)}\otimes\tau_{-1}^{(i)}$ by using the inverse of Bosonization I,
then by using Bosonization II, bigraded Toda hierarchy $L_1^M=L_2^N$ is generalized to GBTH in the form below
\begin{align*}
L_1^M(n)=L_2^{N}(n)+\sum_{j\in \mathbb Z}\sum_{i=1}^{m}q^{(i)}_n\Lambda^jr^{(i)}_{n+1},\quad M, N\geq 1,
\end{align*}
which is equivalent to the following three relations
\begin{align*}
&\Big(L_1^M(n)\Big)_{\leq -N-1}=\sum_{i=1}^{m}q^{(i)}_n\cdot \Lambda^{-N}\Delta^{-1} \cdot r^{(i)}_{n+1},\\
&\Big(L_2^{N}(n)\Big)_{\geq M+1}=-\sum_{i=1}^{m}q^{(i)}_n\cdot \Lambda^M\Delta^{*-1}\cdot r^{(i)}_{n+1},\\
&\Big(L_1^M(n)\Big)_{\geq -N}=\Big(L_2^{N}(n)\Big)_{\leq M}+\sum_{j=-N}^M\sum_{i=1}^{m}q^{(i)}_n\Lambda^jr^{(i)}_{n+1}.
\end{align*}
For GBTH, there are still many questions needing further discussion. For instance, solutions of GBTH remains needing further discussion.

\appendix
\section{}\label{appA}
The aim of this appendix is to prove that
if $\tau\in \mathcal{F}$\ satisfies fermionic KP hierarchy \eqref{eq:fermionic KP hierarchy eq1}, then there exists a unique $\mathcal{F}_k$\ such that $\tau\in\mathcal{F}_k$.
Firstly assume $\tau$\ satisfying fermionic KP hierarchy \eqref{eq:fermionic KP hierarchy eq1}, has the following form
$$\tau=\tau_{i_1}+\tau_{i_2}+\cdots +\tau_{i_\gamma}\in \mathcal{F},\quad 0\neq\tau_{i_\alpha}\in \mathcal{F}_{i_\alpha},\ i_\alpha\neq i_\beta,\ \gamma\geq2.$$
Then by \eqref{eq:fermionic KP hierarchy eq1}, we can know
$$
\sum_{j\in \mathbb{Z}}\sum_{\alpha,\ \beta=1}^{\gamma}\psi_j^+\tau_{i_\alpha}\otimes \psi_{-j}^-\tau_{i_\beta}=0.
$$
Note that\ $\psi_j^+\tau_{i_\alpha }\in \mathcal{F}_{i_\alpha+1}$\ and $\psi_{-j}^-\tau_{i_\beta}\in \mathcal{F}_{i_\beta-1}$, thus\
$\sum\limits_{j\in \mathbb{Z}}\psi_j^+\tau_{i_\alpha}\otimes \psi_{-j}^-\tau_{i_\beta}$ is linear independent for different $i_\alpha$ and
 $i_\beta$, which means
\begin{align*}
\sum_{j\in \mathbb{Z}}\psi_j^+\tau_{i_\alpha}\otimes \psi_{-j}^-\tau_{i_\beta}=0,\quad \forall\ \alpha,\beta=1,2,\cdots \gamma.
\end{align*}

For convenience, let us denote
$\Omega_{(0)}=\sum_{i\in \mathbb{Z}+1/2}\psi_i^{+}\otimes \psi_{-i}^-$.\
In order to prove $\gamma=1$, let us review two facts below. The first one\ \cite{Miwa2000} is that
$$
\Omega_{(0)}(\tau_l\otimes\tau_l)=0,\quad\tau_l\in\mathcal{F}_l
\Leftrightarrow
\tau_l=GL_{\infty}|l\rangle
.
$$
Here $
GL_\infty=\{e^{X_1}\cdots e^{X_k}|X_i=\sum\limits_{j,k\in \mathbb{Z}+1/2} a_{i,jk}\psi^+_j\psi^-_k\}
$\ and $|l\rangle$\ is defined by
\begin{align*}
|l\rangle=\begin{cases}
\psi^{+}_{\frac{1}{2}-l}\psi^{+}_{\frac{3}{2}-l}\cdots\psi^{+}_{-\frac{1}{2}}|0\rangle,&\quad l>0;\\
|0\rangle,&\quad l=0;\\
\psi^{-}_{\frac{1}{2}+l}\psi^{-}_{\frac{3}{2}+l}\cdots\psi^{-}_{-\frac{1}{2}}|0\rangle,&\quad l<0.
\end{cases}
\end{align*}
Another one is given in \cite{Kac1998}. If $\tau_i\in GL_{\infty}|i\rangle$, then
$$
\Omega_{(0)}(\tau_{s+l}\otimes\tau_{s})=0
\Leftrightarrow
(\Omega_{(0)})^l(\tau_{s}\otimes \tau_{s+l})=(-1)^{\frac{l(l-1)}{2}}l!\tau_{s+l}\otimes \tau_{s},\quad l\geq0.$$
Based on above two facts, we can know that when $i_\alpha>i_\beta$,
\begin{align*}
0=(\Omega_{(0)})^{i_\alpha-i_\beta}(\tau_{i_\beta}\otimes \tau_{i_\alpha})=
(-1)^{\frac{(i_\alpha-i_\beta)(i_\alpha-i_\beta-1)}{2}}(i_\alpha-i_\beta)!\tau_{i_\alpha}\otimes \tau_{i_\beta},
\end{align*}
where we have used $\Omega_{(0)}(\tau_{i_\alpha}\otimes \tau_{i_\beta})=0$.
Therefore $\tau_{i_\alpha}=0\ \text{or}\ \tau_{i_\beta}=0$,\ which contradicts assumptions.
In conclusion, there exists unique $\mathcal{F}_k$\ such that $\tau\in\mathcal{F}_k$.\\
\\
\\

\noindent{\bf Acknowledgements}:

This work is supported by National Natural Science Foundation of China (Grant Nos. 12171472 and 12261072)
and ``Qinglan Project" of Jiangsu Universities.\\

\noindent{\bf Conflict of Interest}:

 The authors have no conflicts to disclose.\\

\noindent{\bf Data availability}:

Date sharing is not applicable to this article as no new data were created or analyzed in this study.

\end{document}